\documentclass[11pt]{article}%
\usepackage{amsfonts}
\usepackage{amsmath}
\usepackage{amssymb}
\usepackage{graphicx}
\usepackage{fullpage}
\usepackage{hyperref}%
\providecommand{\U}[1]{\protect\rule{.1in}{.1in}}
\newtheorem{theorem}{Theorem}

\newtheorem{corollary}[theorem]{Corollary}

\newtheorem{definition}[theorem]{Definition}

\newtheorem{lemma}[theorem]{Lemma}

\newtheorem{proposition}[theorem]{Proposition}
\newtheorem{remark}[theorem]{Remark}

\newenvironment{proof}[1][Proof]{\noindent\textbf{#1.} }{\ \rule{0.5em}{0.5em}}
\numberwithin{equation}{section}
\begin{document}

\title{\textbf{Union bound for quantum information processing}}
\author{Samad Khabbazi Oskouei\thanks{Department of Mathematics, Islamic Azad
University, Varamin-Pishva Branch, 33817-7489, Iran}
\and Stefano Mancini\thanks{School of Science and Technology, University of
Camerino, Via M.~delle Carceri 9, I-62032 Camerino, Italy \& INFN--Sezione
Perugia, Via A.~Pascoli, I-06123 Perugia, Italy}
\and Mark M.~Wilde\thanks{Hearne Institute for Theoretical Physics, Department of
Physics and Astronomy, Center for Computation and Technology, Louisiana State
University, Baton Rouge, Louisiana 70803, USA}}
\maketitle

\begin{abstract}
In this paper, we prove a quantum union bound that is relevant when performing
a sequence of binary-outcome quantum measurements on a quantum state. The
quantum union bound proved here involves a tunable parameter that can be
optimized, and this tunable parameter plays a similar role to a parameter
involved in the Hayashi-Nagaoka inequality [\textit{IEEE Trans.~Inf.~Theory},
49(7):1753 (2003)], used often in quantum information theory when analyzing
the error probability of a square-root measurement. An advantage of the proof
delivered here is that it is elementary, relying only on basic properties of
projectors, the Pythagorean theorem, and the Cauchy--Schwarz inequality. As a
non-trivial application of our quantum union bound, we prove that a sequential
decoding strategy for classical communication over a quantum channel achieves
a lower bound on the channel's second-order coding rate. This demonstrates the
advantage of our quantum union bound in the non-asymptotic regime, in which a
communication channel is called a finite number of times. We expect that the
bound will find a range of applications in quantum communication theory,
quantum algorithms, and quantum complexity theory.

\end{abstract}

\section{Introduction}

The union bound, alternatively known as Boole's inequality, represents one of
the simplest yet non-trivial methods for bounding the probability that either
one event or another occurs, in terms of the probabilities of the individual
events (see, e.g., \cite{R77}). By induction, the bound applies to the union
of multiple events, and it often provides a good enough bound in a variety of
applications whenever the probabilities of the individual events are small
relative to the number of events. Concretely, given a finite set
$\{A_{i}\}_{i=1}^{L}$ of events, the union bound is the following inequality:%
\begin{equation}
\Pr\!\left\{  \bigcup\limits_{i=1}^{L}A_{i}\right\}  \leq\sum_{i=1}^{L}%
\Pr\{A_{i}\}.
\end{equation}
By applying DeMorgan's law and basic rules of probability theory, we can
rewrite the union bound such that it applies to the probability that an
intersection of events does not occur%
\begin{equation}
1-\Pr\left\{  \bigcap\limits_{i=1}^{L}A_{i}\right\}  \leq\sum_{i=1}^{L}%
\Pr\{A_{i}^{c}\},
\end{equation}
and this is the form in which it is typically employed in applications.
Recently, the union bound has been listed as the second step to try when
attempting to \textquotedblleft upper-bound the probability of something
bad,\textquotedblright\ with the first step being to determine if the trivial
bound of one is reasonable in a given application \cite{A18}.

Generalizing the union bound to a quantum-mechanical setup is non-trivial. A
natural setting in which we would consider this generalization is when the
goal is to bound the probability that two or more successive measurement
outcomes do not occur. Concretely, suppose that the state of a quantum system
is given by a density operator $\rho$. Suppose that there are $L$ projective
quantum measurements $\{P_{i},I-P_{i}\}$ for $i\in\{1,\ldots,L\}$, where
$P_{i}$ is a projector, thus satisfying $P_{i}=P_{i}^{\dag}$ and $P_{i}%
=P_{i}^{2}$ by definition. Suppose that the first measurement is performed,
followed by the second measurement, and so on. If the projectors $P_{1}%
,\ldots,P_{L}$ commute, then the probability that the outcomes $P_{1}%
,\ldots,P_{L}$ do not occur is calculated by applying the Born rule and can be
bounded as%
\begin{equation}
1-\operatorname{Tr}\{P_{L}P_{L-1}\cdots P_{1}\rho P_{1}\cdots P_{L-1}%
\}\leq\sum_{i=1}^{L}\operatorname{Tr}\{(I-P_{i})\rho\},
\label{eq:classical-ub}%
\end{equation}
with the bound following essentially from an application of the union bound.
However, if the projectors $P_{1},\ldots,P_{L}$ do not commute, then classical
reasoning does not apply and alternative methods are required.

Recently, Gao proved a quantum union bound \cite{G15}\ that has been useful in
a variety of applications, including quantum communication theory
\cite{G15,H15,BCN16,R17,YSM17}, quantum algorithms \cite{MW16,A16,HYM17},
quantum complexity theory \cite{A16,WHKN18}, and Hamiltonian complexity theory
\cite{AAV16,KL18}. Given an arbitrary set of projectors $\{P_{i}\}_{i=1}^{L}$,
each corresponding to one outcome of a binary-valued measurement, Gao's
quantum union bound is the following inequality \cite[Theorem 1]{G15}:%
\begin{equation}
1-\operatorname{Tr}\{P_{L}P_{L-1}\cdots P_{1}\rho P_{1}\cdots P_{L-1}%
\}\leq4\sum_{i=1}^{L}\operatorname{Tr}\{(I-P_{i})\rho\}. \label{eq:gao-qub}%
\end{equation}
By comparing \eqref{eq:gao-qub} with \eqref{eq:classical-ub}, we notice that
the only difference is the factor of four in \eqref{eq:gao-qub}. The factor of
four is inconsequential for many applications, but nevertheless, it is natural
to wonder whether this bound can be improved. Furthermore, at least one
application in which improving the factor of four does make a difference is in
the context of whether a sequential decoding strategy can be used to achieve
the second-order coding rate for classical communication---we discuss this
application in more detail later.

\section{Summary of results}

In this paper, we prove the following quantum union bound:

\begin{theorem}
[Quantum union bound]\label{thm:main-result}Let $\rho$ be a density operator
acting on a separable Hilbert space $\mathcal{H}$, let $\{P_{i}\}_{i=1}^{L}$
be an arbitrary set of projectors, each acting on $\mathcal{H}$, and let $c>0$
be an arbitrary positive constant. Then%
\begin{multline}
1-\operatorname{Tr}\{P_{L}P_{L-1}\cdots P_{1}\rho P_{1}\cdots P_{L-1}%
\}\leq\left(  1+c\right)  \operatorname{Tr}\{\left(  I-P_{L}\right)  \rho\}\\
+\left(  2+c+c^{-1}\right)  \sum_{i=2}^{L-1}\operatorname{Tr}\{(I-P_{i}%
)\rho\}+\left(  2+c^{-1}\right)  \operatorname{Tr}\{\left(  I-P_{1}\right)
\rho\}.
\end{multline}

\end{theorem}

Our proof of the above theorem is elementary, relying only on basic properties
of projectors, the Pythagorean theorem, and the Cauchy--Schwarz inequality.
Furthermore, the theorem directly applies to states of infinite-dimensional
quantum systems and can thus be employed to analyze practical situations
involving not only qubits but also bosonic quantum systems \cite{S17}. Similar
to the classical case discussed in the introduction, the quantum union bound
of Theorem~\ref{thm:main-result} provides a useful bound when the individual
probabilities $\operatorname{Tr}\{(I-P_{i})\rho\}$ are small relative to the
number~$L$ of them, and this scenario occurs, for example, in the application
to communication presented in Section~\ref{sec:ea-comm}. Furthermore, the
tunable parameter $c>0$ is a significant advantage of our quantum union bound,
and it is essential in the application mentioned above, in which it really is
necessary for $c>0$ to be decreasing with the number of channel uses so that
the prefactor in front of the term $\operatorname{Tr}\{\left(  I-P_{L}\right)
\rho\}$ is as close to one as possible. More generally, one could certainly
take an infimum over the parameter $c>0$ in any given application in order to
have the upper bound be as tight as possible.

Our quantum union bound represents a strict improvement over that of Gao's in
\eqref{eq:gao-qub}. Indeed, by setting $c=1$ and then loosening the above
bound further, we recover Gao's. Our quantum union bound can also be compared
with the Hayashi--Nagaoka (HN) inequality from \cite[Lemma~2]{HN03}, which is
often used to analyze the error probability of the square-root measurement.
The HN inequality also features a tunable parameter $c>0$, and this is one of
the main reasons why quantum information theory has recently advanced in the
direction of characterizing second-order asymptotics for communication tasks
\cite{TH12,TT13,datta2015second,DTW14,WRG15,TBR15,L16,WTB16,wilde2017position}%
. Our quantum union bound provides essentially the same trade-off given by the
HN\ inequality, but just slightly improved, in the sense that the prefactor
for the term $\operatorname{Tr}\{\left(  I-P_{L}\right)  \rho\}$ is $1+c$,
while the prefactor for $L-2$ other terms is $2+c+c^{-1}$ and the prefactor
for the term $\operatorname{Tr}\{\left(  I-P_{1}\right)  \rho\}$ is $2+c^{-1}%
$, the last prefactor representing the improvement.

In the previous paragraphs, we focused exclusively on the comparison of
Theorem~\ref{thm:main-result} with Gao's bound in \eqref{eq:gao-qub}. However,
there were other works that preceded Gao's, which we recall now. \cite{A06}
established a quantum union bound, with applications in quantum complexity
theory. \cite{PhysRevA.85.012302} analyzed the error probability of a
sequential decoding strategy and proved that it can achieve the Holevo
information of a quantum channel for classical communication. The work of
\cite{PhysRevA.85.012302} then inspired \cite{S11}, who established another
quantum union bound (also called \textquotedblleft non-commutative union
bound\textquotedblright) of the following form:%
\begin{equation}
1-\operatorname{Tr}\{P_{L}P_{L-1}\cdots P_{1}\rho P_{1}\cdots P_{L-1}%
\}\leq2\sqrt{\sum_{i=1}^{L}\operatorname{Tr}\{(I-P_{i})\rho\}}.
\label{eq:sen-bnd}%
\end{equation}
\cite{Wilde20130259} subsequently generalized the result of \cite{S11} beyond
projectors, such that it would hold for a set of operators $\{\Lambda
_{i}\}_{i=1}^{L}$, each of which satisfies $0\leq\Lambda_{i}\leq I$. Then
Gao's bound in \eqref{eq:gao-qub} appeared after \cite{Wilde20130259}.
Clearly, Gao's bound was a significant improvement over \eqref{eq:sen-bnd},
eliminating the square root at the cost of a doubling of the prefactor.

To demonstrate an application in which Theorem~\ref{thm:main-result} is
useful, we show how a sequential decoding strategy achieves a lower bound on
the second-order coding rate for classical communication over a quantum
channel. We consider the cases in which there is entanglement assistance\ as
well as no assistance, and our result here also covers the important case when
the channel takes input density operators acting on a separable Hilbert space
to output density operators acting on a separable Hilbert space. An advantage
of our proof is that it is arguably simpler than other approaches that could
be taken to solve this problem, relying on a method called position-based
coding \cite{anshu2017one}, as well as sequential decoding
\cite{PhysRevA.85.012302,S11,Wilde20130259}, and an error analysis that uses
Theorem~\ref{thm:main-result}. Our proof can be compared with the proof from
\cite{H03,H04}, in which it was shown how to achieve the capacity for
energy-constrained classical communication (i.e., the first-order coding
rate), and we advocate here that our proof is considerably simpler.

We organize the rest of our paper as follows. In Section~\ref{sec:proof}, we
provide a proof of Theorem~\ref{thm:main-result}.
In Section~\ref{sec:gen-to-POVMs}, we consider the generalization of
Theorem~\ref{thm:main-result} to positive operator-valued measures (POVMs).
Section~\ref{sec:ea-comm} discusses the application to obtaining a lower bound
on the second-order coding rate for classical communication. In
Section~\ref{sec:concl}, we conclude with a summary and discuss some open
directions for future research.

\section{Proof of Theorem~\ref{thm:main-result}}

\label{sec:proof}

We prove our main result, Theorem~\ref{thm:main-result}, by establishing the
following more general result:

\begin{theorem}
\label{T1}Let $\mathcal{H}$ be a separable Hilbert space, let $|\psi\rangle
\in\mathcal{H}$, let $\{P_{i}\}_{i=1}^{L}$ be a finite set of projectors
acting on $\mathcal{H}$, and let $c>0$. Then%
\begin{multline}
\left\Vert |\psi\rangle\right\Vert _{2}^{2}-\left\Vert P_{L}P_{L-1}\cdots
P_{1}|\psi\rangle\right\Vert _{2}^{2}\leq(1+c)\left\Vert \left(
I-P_{L}\right)  |\psi\rangle\right\Vert _{2}^{2}\label{asbound}\\
+(2+c+c^{-1})\sum_{i=2}^{L-1}\left\Vert \left(  I-P_{i}\right)  |\psi
\rangle\right\Vert _{2}^{2}+(2+c^{-1})\left\Vert \left(  I-P_{1}\right)
|\psi\rangle\right\Vert _{2}^{2}.
\end{multline}

\end{theorem}

Theorem~\ref{thm:main-result} is a direct consequence of Theorem~\ref{T1}.
Indeed, a density operator $\rho$ acting on a separable Hilbert space has a
spectral decomposition as follows:%
\begin{equation}
\rho=\sum_{j\in\mathcal{J}}p_{j}|\psi_{j}\rangle\langle\psi_{j}|,
\end{equation}
where the index set $\mathcal{J}$ is countable, $\{p_{j}\}_{j\in\mathcal{J}}$
is a probability distribution, and $\{|\psi_{j}\rangle\}_{j\in\mathcal{J}}$ is
an orthonormal set of eigenvectors \cite{Hol11}. Applying Theorem~\ref{T1}, we
find that%
\begin{align}
&  1-\operatorname{Tr}\{P_{L}P_{L-1}\cdots P_{1}|\psi_{j}\rangle\langle
\psi_{j}|P_{1}\cdots P_{L-1}\}\nonumber\\
&  =\left\Vert |\psi_{j}\rangle\right\Vert _{2}^{2}-\left\Vert P_{L}%
P_{L-1}\cdots P_{1}|\psi_{j}\rangle\right\Vert _{2}^{2}\\
&  \leq(1+c)\left\Vert \left(  I-P_{L}\right)  |\psi_{j}\rangle\right\Vert
_{2}^{2}+(2+c+c^{-1})\sum_{i=2}^{L-1}\left\Vert \left(  I-P_{i}\right)
|\psi_{j}\rangle\right\Vert _{2}^{2}\nonumber\\
&  \qquad\qquad+\left(  2+c^{-1}\right)  \left\Vert \left(  I-P_{1}\right)
|\psi_{j}\rangle\right\Vert _{2}^{2}\\
&  =\left(  1+c\right)  \operatorname{Tr}\{\left(  I-P_{L}\right)  |\psi
_{j}\rangle\langle\psi_{j}|\}+\left(  2+c+c^{-1}\right)  \sum_{i=2}%
^{L-1}\operatorname{Tr}\{(I-P_{i})|\psi_{j}\rangle\langle\psi_{j}%
|\}\nonumber\\
&  \qquad\qquad\left(  2+c^{-1}\right)  +\operatorname{Tr}\{\left(
I-P_{1}\right)  |\psi_{j}\rangle\langle\psi_{j}|\}.
\end{align}
The reduction from Theorem~\ref{thm:main-result} to Theorem~\ref{T1} follows
by averaging over the distribution $\{p_{j}\}_{j\in\mathcal{J}}$.

So now we shift our focus to proving Theorem~\ref{T1}, and we do so with the
aid of several lemmas. To simplify the notation, hereafter we employ the
following shorthand:%
\begin{align}
\left\Vert \cdots\right\Vert  &  \equiv\left\Vert \cdots|\psi\rangle
\right\Vert _{2},\label{eq:shorthand-1}\\
\langle\cdots\rangle &  \equiv\langle\psi|\cdots|\psi\rangle,\\
Q_{i}  &  \equiv I-P_{i}. \label{eq:q_i}%
\end{align}
The convention we take with the shorthand $\langle A\rangle$ for a
non-Hermitian operator $A$ is that $\langle A\rangle= \langle\psi\vert
\varphi\rangle$ where $\vert\varphi\rangle= A\vert\psi\rangle$. Furthermore,
we also assume without loss of generality that the vector $|\psi\rangle$ in
Theorem~\ref{T1} is a unit vector. Clearly, this assumption can be easily
released by scaling the resulting inequality by an arbitrary positive number.

First recall that, due to the idempotence of projectors, we have the following
identities holding for all $i\in\{1,2,\ldots,L\}$:%
\begin{equation}
\langle Q_{i}P_{i-1}\cdots P_{1}\rangle=\langle Q_{i}Q_{i}P_{i-1}\cdots
P_{1}\rangle,\qquad\langle P_{1}\cdots P_{i}\rangle=\langle P_{1}\cdots
P_{i}P_{i}\rangle, \label{projection}%
\end{equation}
under the convention that $P_{i-1}\cdots P_{1}=P_{1}\cdots P_{i-1}=I$ for
$i=1$.

\begin{lemma}
\label{L1}For a set $\{P_{i}\}_{i=1}^{L}$ of projectors acting on a separable
Hilbert space $\mathcal{H}$, a unit vector $|\psi\rangle\in\mathcal{H}$, and
employing the shorthand in \eqref{eq:shorthand-1}--\eqref{eq:q_i}, we have the
following identities:%
\begin{align}
\sum_{i=1}^{L}\langle Q_{i}P_{i-1}\cdots P_{1}\rangle &  =1-\langle
P_{L}\cdots P_{1}\rangle,\label{pro1}\\
\sum_{i=1}^{L}\langle P_{1}\cdots P_{i-1}Q_{i}\rangle &  =1-\langle
P_{1}\cdots P_{L}\rangle,\label{pro2}\\
\sum_{i=1}^{L}\langle P_{1}\cdots P_{i-1}Q_{i}P_{i-1}\cdots P_{1}\rangle &
=1-\langle P_{1}\cdots P_{L}\cdots P_{1}\rangle,\label{pro3}\\
1-\sqrt{\langle P_{L}\rangle}\sqrt{\langle P_{1}\cdots P_{L}\cdots
P_{1}\rangle}  &  \leq\sum_{i=1}^{L}\sqrt{\langle Q_{i}\rangle}\sqrt{\langle
P_{1}\cdots P_{i-1}Q_{i}P_{i-1}\cdots P_{1}\rangle}, \label{pro4}%
\end{align}
under the convention that $P_{i-1}\cdots P_{1}=P_{1}\cdots P_{i-1}=I$ for
$i=1$.
\end{lemma}

\begin{proof}
The following identities are straightforward:%
\begin{align}
1  &  =\langle Q_{1}\rangle+\langle Q_{2}P_{1}\rangle+\cdots+\langle
Q_{L-1}P_{L-2}\cdots P_{1}\rangle+\langle Q_{L}P_{L-1}\cdots P_{1}%
\rangle+\langle P_{L}P_{L-1}\cdots P_{1}\rangle,\label{rel1}\\
1  &  =\langle Q_{1}\rangle+\langle P_{1}Q_{2}\rangle+\cdots+\langle
P_{1}\cdots P_{L-2}Q_{L-1}\rangle+\langle P_{1}\cdots P_{L-1}Q_{L}%
\rangle+\langle P_{1}\cdots P_{L-1}P_{L}\rangle,\label{rel2}\\
1  &  =\langle Q_{1}\rangle+\langle P_{1}Q_{2}P_{1}\rangle+\cdots+\langle
P_{1}\cdots P_{L-2}Q_{L-1}P_{L-2}\cdots P_{1}\rangle\nonumber\\
&  \qquad\qquad+\langle P_{1}\cdots P_{L-1}Q_{L}P_{L-1}\cdots P_{1}%
\rangle+\langle P_{1}\cdots P_{L-1}P_{L}P_{L-1}\cdots P_{1}\rangle.
\label{rel3}%
\end{align}
Consequently, from the equalities in \eqref{rel1}, \eqref{rel2}, and
\eqref{rel3}, we obtain \eqref{pro1}, \eqref{pro2}, and \eqref{pro3},
respectively. The following equality is a direct consequence of \eqref{rel1}
and \eqref{projection}:%
\begin{multline}
1=\langle Q_{1}\rangle+\langle Q_{2}Q_{2}P_{1}\rangle+\cdots+\langle
Q_{L-1}Q_{L-1}P_{L-2}\cdots P_{1}\rangle+\langle Q_{L}Q_{L}P_{L-1}\cdots
P_{1}\rangle\label{rel4}\\
+\langle P_{L}P_{L}P_{L-1}\cdots P_{1}\rangle.
\end{multline}
By applying the Cauchy-Schwarz inequality to \eqref{rel4}, we find that%
\begin{multline}
1\leq\langle Q_{1}\rangle+\sqrt{\langle Q_{2}\rangle}\sqrt{\langle P_{1}%
Q_{2}P_{1}\rangle}+\cdots+\sqrt{\langle Q_{L}\rangle}\sqrt{\langle P_{1}\cdots
P_{L-1}Q_{L}P_{L-1}\cdots P_{1}\rangle}\\
+\sqrt{\langle P_{L}\rangle}\sqrt{\langle P_{1}\cdots P_{L-1}P_{L}%
P_{L-1}\cdots P_{1}\rangle},
\end{multline}
from which \eqref{pro4} immediately follows.
\end{proof}

\begin{lemma}
\label{L3}For a set $\{P_{i}\}_{i=1}^{L}$ of projectors acting on a separable
Hilbert space $\mathcal{H}$, a unit vector $|\psi\rangle\in\mathcal{H}$, and
employing the shorthand in \eqref{eq:shorthand-1}--\eqref{eq:q_i}, the
following inequality holds for $L\geq2$:%
\begin{equation}
\sum_{i=1}^{L}\left\Vert Q_{i}(I-P_{i-1}\cdots P_{1})\right\Vert ^{2}\leq
\sum_{i=1}^{L-1}\left\Vert Q_{i}\right\Vert ^{2},
\end{equation}
under the convention that $P_{i-1}\cdots P_{1}=P_{1}\cdots P_{i-1}=I$ for
$i=1$. Equivalently,%
\begin{equation}
\sum_{i=2}^{L}\left\Vert Q_{i}(I-P_{i-1}\cdots P_{1})\right\Vert ^{2}\leq
\sum_{i=1}^{L-1}\left\Vert Q_{i}\right\Vert ^{2},
\end{equation}
due to the aforementioned convention.
\end{lemma}

\begin{proof}
Consider the following chain of equalities:%
\begin{align}
&  \sum_{i=1}^{L}\left\Vert Q_{i}(I-P_{i-1}\cdots P_{1})\right\Vert ^{2}%
=\sum_{i=1}^{L}\left\Vert Q_{i}-Q_{i}P_{i-1}\cdots P_{1}\right\Vert
^{2}\nonumber\\
&  =\sum_{i=1}^{L}\left(  \left\Vert Q_{i}\right\Vert ^{2}-\langle
Q_{i}P_{i-1}\cdots P_{1}\rangle-\langle P_{1}\cdots P_{i-1}Q_{i}%
\rangle+\langle P_{1}\cdots P_{i-1}Q_{i}P_{i-1}\cdots P_{1}\rangle\right)
\label{apro1}\\
&  =\left(  \sum_{i=1}^{L}\left\Vert Q_{i}\right\Vert ^{2}\right)  -1+\langle
P_{L}\cdots P_{1}\rangle-1+\langle P_{1}\cdots P_{L}\rangle+1-\langle
P_{1}\cdots P_{L}\cdots P_{1}\rangle\label{apro2}\\
&  =\left(  \sum_{i=1}^{L}\left\Vert Q_{i}\right\Vert ^{2}\right)  -1+\langle
P_{L}P_{L}P_{L-1}\cdots P_{1}\rangle+\langle P_{1}\cdots P_{L-1}P_{L}%
P_{L}\rangle-\langle P_{1}\cdots P_{L}\cdots P_{1}\rangle.
\label{eq:connector}%
\end{align}
To obtain \eqref{apro1}, we used the identities in \eqref{projection}. Next,
to get \eqref{apro2}, the identities in~\eqref{pro1}, \eqref{pro2},
and~\eqref{pro3} of Lemma~\ref{L1} were used. Continuing, we have that%
\begin{align}
\text{Eq.}~\eqref{eq:connector}  &  \leq\left(  \sum_{i=1}^{L}\left\Vert
Q_{i}\right\Vert ^{2}\right)  -1-\langle P_{1}\cdots P_{L}\cdots P_{1}%
\rangle+2\sqrt{\langle P_{L}\rangle}\sqrt{\langle P_{1}\cdots P_{L}\cdots
P_{1}\rangle}\label{apro3}\\
&  =\left(  \sum_{i=1}^{L}\left\Vert Q_{i}\right\Vert ^{2}\right)  -1+\langle
P_{L}\rangle-\left(  \sqrt{\langle P_{L}\rangle}-\sqrt{\langle P_{1}\cdots
P_{L}\cdots P_{1}\rangle}\right)  ^{2}\\
&  \leq\left(  \sum_{i=1}^{L}\left\Vert Q_{i}\right\Vert ^{2}\right)
-\left\Vert Q_{L}\right\Vert ^{2}=\sum_{i=1}^{L-1}\left\Vert Q_{i}\right\Vert
^{2}.
\end{align}
To obtain \eqref{apro3}, the Cauchy-Schwarz inequality was employed.
\end{proof}

\bigskip

We are now in a position to prove Theorem~\ref{T1}:

\bigskip

\begin{proof}
[Proof of Theorem~\ref{T1}]Consider that%
\begin{align}
1-\Vert P_{L}\cdots P_{1}\Vert^{2}  &  =1-\langle P_{1}\cdots P_{L}\cdots
P_{1}\rangle+2\left(  1-\sqrt{\langle P_{L}\rangle}\sqrt{\langle P_{1}\cdots
P_{L}\cdots P_{1}\rangle}\right) \nonumber\\
&  \qquad\qquad-2\left(  1-\sqrt{\langle P_{L}\rangle}\sqrt{\langle
P_{1}\cdots P_{L}\cdots P_{1}\rangle}\right) \\
&  =2\left(  1-\sqrt{\langle P_{L}\rangle}\sqrt{\langle P_{1}\cdots
P_{L}\cdots P_{1}\rangle}\right) \nonumber\\
&  \qquad\qquad-\left(  \sqrt{\langle P_{L}\rangle}-\sqrt{\langle P_{1}\cdots
P_{L}\cdots P_{1}\rangle}\right)  ^{2}-1+\langle P_{L}\rangle.
\label{eq:connector-3}%
\end{align}
Continuing, we have that%
\begin{align}
\text{Eq.}~\eqref{eq:connector-3}  &  \leq-\left\Vert Q_{L}\right\Vert
^{2}+2\left(  1-\sqrt{\langle P_{L} \rangle}\sqrt{\langle P_{1}\cdots
P_{L}\cdots P_{1}\rangle}\right) \label{mpro1}\\
&  \leq-\left\Vert Q_{L}\right\Vert ^{2}+2\sum_{i=1}^{L}\sqrt{\langle
Q_{i}\rangle}\sqrt{\langle P_{1}\cdots P_{i-1}Q_{i}P_{i-1}\cdots P_{1}\rangle
}\label{mpro2}\\
&  \leq-\left\Vert Q_{L}\right\Vert ^{2}+2\sum_{i=1}^{L}\sqrt{\langle
Q_{i}\rangle}\left(  \left\Vert Q_{i}\right\Vert +\left\Vert Q_{i}%
(I-P_{i-1}\cdots P_{1})\right\Vert \right)  . \label{mpro3}%
\end{align}
First, \eqref{mpro1} is obtained by observing that
\begin{equation}
-\left(  \sqrt{\langle P_{L}\rangle}-\sqrt{\langle P_{1}\cdots P_{L}\cdots
P_{1}\rangle}\right)  ^{2}-1+\langle P_{L}\rangle\leq-1+\langle P_{L}%
\rangle=-\Vert Q_{L}\Vert^{2}.
\end{equation}
Next, \eqref{mpro2} follows from~\eqref{pro4} of Lemma~\ref{L1}. Then,
\eqref{mpro3} is a consequence of the triangle inequality:
\begin{align}
\sqrt{\langle P_{1}\cdots P_{i-1}Q_{i}P_{i-1}\cdots P_{1}\rangle}  &  = \Vert
Q_{i}P_{i-1}\cdots P_{1} \Vert\\
&  = \Vert Q_{i}(-I + I- P_{i-1}\cdots P_{1})\Vert\\
&  \leq\Vert Q_{i}\Vert+ \Vert Q_{i} (I- P_{i-1}\cdots P_{1})\Vert,
\end{align}
under the convention that $P_{i-1}\cdots P_{1}=I$ for $i=1$. Continuing, we
have that%
\begin{align}
\text{Eq.}~\eqref{mpro3}  &  =-\left\Vert Q_{L}\right\Vert ^{2}+2\sum
_{i=1}^{L}\left\Vert Q_{i}\right\Vert ^{2}+2\sum_{i=1}^{L}\left(  \left\Vert
Q_{i}\right\Vert \left\Vert Q_{i}(I-P_{i-1}\cdots P_{1})\right\Vert \right) \\
&  =-\left\Vert Q_{L}\right\Vert ^{2}+2\sum_{i=1}^{L}\left\Vert Q_{i}%
\right\Vert ^{2}+2\sum_{i=2}^{L}\left(  \left\Vert Q_{i}\right\Vert \left\Vert
Q_{i}(I-P_{i-1}\cdots P_{1})\right\Vert \right) \label{eq:apply-conv}\\
&  \leq-\left\Vert Q_{L}\right\Vert ^{2}+2\sum_{i=1}^{L}\left\Vert
Q_{i}\right\Vert ^{2}+\sum_{i=2}^{L}\left(  c\left\Vert Q_{i}\right\Vert
^{2}+c^{-1}\left\Vert Q_{i}(I-P_{i-1}\cdots P_{1})\right\Vert ^{2}\right)
\label{mpro3bis}\\
&  \leq-\left\Vert Q_{L}\right\Vert ^{2}+2\sum_{i=1}^{L}\left\Vert
Q_{i}\right\Vert ^{2}+c\sum_{i=2}^{L}\left\Vert Q_{i}\right\Vert ^{2}%
+c^{-1}\sum_{i=1}^{L-1}\left\Vert Q_{i}\right\Vert ^{2}\label{mpro4}\\
&  \leq(1+c)\left\Vert Q_{L}\right\Vert ^{2}+(2+c^{-1})\left\Vert
Q_{1}\right\Vert ^{2}+(2+c+c^{-1})\sum_{i=2}^{L-1}\left\Vert Q_{i}\right\Vert
^{2}. \label{mpro5}%
\end{align}
Eq.~\eqref{eq:apply-conv} follows from the convention that $P_{i-1}\cdots
P_{1}=I$ for $i=1$. Eq.~\eqref{mpro3bis} is a consequence of the inequality
$2xy\leq cx^{2}+c^{-1}y^{2}$, holding for $x,y\in\mathbb{R}$ and $c>0$.
Finally, \eqref{mpro4} is obtained by using Lemma \ref{L3}.
\end{proof}

\section{Generalization to POVMs}

\label{sec:gen-to-POVMs}

Just as the bound from \cite{S11} was generalized in \cite[Section~3]%
{Wilde20130259}\ from projectors to positive semi-definite operators having
eigenvalues between zero and one, we can do the same here. This generalization
is useful for applications, and we discuss one such application in the next section.

We now give an extension of the quantum union bound in
Theorem~\ref{thm:main-result}\ that applies for general measurements. The main
idea behind it is the well known Naimark extension theorem, following the
approach from \cite[Section~3]{Wilde20130259}.

\begin{lemma}
\label{lem:non-comm-bound}Let $\rho$ be a positive semi-definite operator
acting on a separable Hilbert space $\mathcal{H}_{S}$, let $\{\Lambda
_{i}\}_{i=1}^{L}$ denote a set of positive semi-definite operators such that
$0\leq\Lambda_{i}\leq I$ for all $i\in\left\{  1,\ldots,L\right\}  $, and let
$c>0$. Then the following quantum union bound holds%
\begin{multline}
\operatorname{Tr}\{\rho\}-\operatorname{Tr}\{\Pi_{\Lambda_{L}}\cdots
\Pi_{\Lambda_{1}}(\rho\otimes|\overline{0}\rangle\langle\overline{0}|_{P^{L}%
})\Pi_{\Lambda_{1}}\cdots\Pi_{\Lambda_{L}}\}\leq\left(  1+c\right)
\operatorname{Tr}\{(I-\Lambda_{L})\rho\}\label{eq:ext-main-bound}\\
+\left(  2+c+c^{-1}\right)  \sum_{i=2}^{L-1}\operatorname{Tr}\left\{  \left(
I-\Lambda_{i}\right)  \sigma\right\}  +\left(  2+c^{-1}\right)
\operatorname{Tr}\{(I-\Lambda_{1})\rho\},
\end{multline}
where $|\overline{0}\rangle_{P^{L}}\equiv\left\vert 0\right\rangle _{P_{1}%
}\otimes\cdots\otimes\left\vert 0\right\rangle _{P_{L}}$ is an auxiliary state
of $L$ qubit probe systems and $\Pi_{\Lambda_{i}}$ is a projector defined as
$\Pi_{\Lambda_{i}}\equiv U_{i}^{\dag}P_{i}U_{i}$, for some unitary $U_{i}$ and
projector $P_{i}$ such that%
\begin{equation}
\operatorname{Tr}\{\Pi_{\Lambda_{i}}(\rho\otimes|\overline{0}\rangle
\langle\overline{0}|_{P^{L}})\}=\operatorname{Tr}\{ \Lambda_{i}\rho\} .
\end{equation}

\end{lemma}

\begin{proof}
This extension of Theorem~\ref{thm:main-result}\ follows easily by employing
the Naimark extension theorem. Concretely, to each operator $\Lambda_{i}$, we
associate the following unitary:%
\begin{equation}
U_{SP_{i}}\equiv\sqrt{I_{S}-(\Lambda_{i})_{S}}\otimes\left[  |0\rangle
\langle0|_{P_{i}}+|1\rangle\langle1|_{P_{i}}\right]  +\sqrt{(\Lambda_{i})_{S}%
}\otimes\left[  |1\rangle\langle0|_{P_{i}}-|0\rangle\langle1|_{P_{i}}\right]
.
\end{equation}
Then defining the projectors $\Pi_{\Lambda_{i}}\equiv U_{SP_{i}}^{\dag}\left(
I_{S}\otimes|1\rangle\langle1|_{P_{i}}\right)  U_{SP_{i}}$, a straightforward
calculation gives that $\operatorname{Tr}\{\Pi_{\Lambda_{i}}(\rho
\otimes|\overline{0}\rangle\langle\overline{0}|_{P^{L}})\}=\text{Tr}\{
\Lambda_{i}\rho_{S}\}$. Observe that the operator $\Pi_{\Lambda_{i}}$ is an
orthogonal projector (because it is Hermitian and idempotent), so that
Theorem~\ref{thm:main-result}\ applies to each of these operators. Then
\eqref{eq:ext-main-bound} follows.
\end{proof}

\section{Lower bound on the second-order coding rate for classical
communication}

\label{sec:ea-comm}One application of our main result,
Theorem~\ref{thm:main-result}, is in achieving the second-order coding rate
for classical communication. As we stated earlier, this area of quantum
information theory has advanced in recent years
\cite{TH12,TT13,datta2015second,DTW14,WRG15,TBR15,L16,WTB16,wilde2017position}%
, with one of the main reasons being the availability of the tunable parameter
$c>0$ in the Hayashi--Nagaoka inequality \cite[Lemma~2]{HN03}. That is, one
can let $c>0$ vary, to become closer to zero, as the number of channel uses increases.

An advantage of our Theorem~\ref{thm:main-result} is that it applies directly
to the case of states and projectors that act on an infinite-dimensional,
separable Hilbert space. Thus, the theorem can be applied directly in order to
achieve a lower bound on the second-order coding rate for classical
communication. To our knowledge, prior to our work here, \cite{WRG15}
presented the only case in which lower bounds on the second-order coding rates
have been considered in this general case, and there, the analysis was limited
to channels that accept a classical input and output a pure quantum state. The
situation that we analyze here is thus more general.

\subsection{Information quantities}

Before we begin with the application, let us recall some information
quantities that are essential in the analysis. Let $\mathcal{H}$ denote a
separable Hilbert space, and let $\mathcal{D}(\mathcal{H})$ denote the set of
density operators acting on $\mathcal{H}$ (positive, semi-definite operators
with trace equal to one). Let spectral decompositions of $\rho,\sigma
\in\mathcal{D}(\mathcal{H})$ be given as%
\begin{equation}
\rho=\sum_{x\in\mathcal{X}}\lambda_{x}P_{x},\qquad\sigma=\sum_{y\in
\mathcal{Y}}\mu_{y}Q_{y}, \label{eq:spec-decomp-rho-sig}%
\end{equation}
where $\mathcal{X}$ and $\mathcal{Y}$ are countable index sets, $\{\lambda
_{x}\}_{x\in\mathcal{X}}$ and $\{\mu_{y}\}_{y\in\mathcal{Y}}$ are probability
distributions with $\lambda_{x},\mu_{y}\geq0$ for all $x\in\mathcal{X}$ and
$y\in\mathcal{Y}$ and $\sum_{x\in\mathcal{X}}\lambda_{x}=\sum_{y\in
\mathcal{Y}}\mu_{y}=1$, and $\{P_{x}\}_{x\in\mathcal{X}}$ and $\{Q_{y}%
\}_{y\in\mathcal{Y}}$ are sets of projections such that $\sum_{x\in
\mathcal{X}}P_{x}=\sum_{y\in\mathcal{Y}}Q_{y}=I$.

The hypothesis testing relative entropy $D_{H}^{\varepsilon}(\rho\Vert\sigma)$
is defined for $\varepsilon\in\left[  0,1\right]  $ as \cite{BD10,WR12}%
\begin{equation}
D_{H}^{\varepsilon}(\rho\Vert\sigma)\equiv-\log_{2}\inf_{\Lambda}\left\{
\operatorname{Tr}\{\Lambda\sigma\}:\operatorname{Tr}\{\Lambda\rho
\}\geq1-\varepsilon\wedge0\leq\Lambda\leq I\right\}  .
\end{equation}
The quantum relative entropy \cite{Lindblad1973}, the quantum relative entropy
variance \cite{TH12,li12,KW17a}, and the $T$ quantity \cite{TH12,li12,KW17a}
are defined as%
\begin{align}
D(\rho\Vert\sigma)  &  \equiv\sum_{x\in\mathcal{X},y\in\mathcal{Y}}\lambda
_{x}\operatorname{Tr}\{P_{x}Q_{y}\}\log_{2}\!\left(  \frac{\lambda_{x}}%
{\mu_{y}}\right)  ,\label{eq:rel-ent-def}\\
V(\rho\Vert\sigma)  &  \equiv\sum_{x\in\mathcal{X},y\in\mathcal{Y}}\lambda
_{x}\operatorname{Tr}\{P_{x}Q_{y}\}\left[  \log_{2}\!\left(  \frac{\lambda
_{x}}{\mu_{y}}\right)  -D(\rho\Vert\sigma)\right]  ^{2}%
,\label{eq:rel-ent-var-def}\\
T(\rho\Vert\sigma)  &  \equiv\sum_{x\in\mathcal{X},y\in\mathcal{Y}}\lambda
_{x}\operatorname{Tr}\{P_{x}Q_{y}\}\left\vert \log_{2}\!\left(  \frac
{\lambda_{x}}{\mu_{y}}\right)  -D(\rho\Vert\sigma)\right\vert ^{3}.
\label{eq:rel-ent-T-def}%
\end{align}
For states $\rho$ and $\sigma$\ satisfying%
\begin{equation}
D(\rho\Vert\sigma),V(\rho\Vert\sigma),T(\rho\Vert\sigma)<\infty,\qquad\qquad
V(\rho\Vert\sigma)>0,
\end{equation}
the following expansion holds for the hypothesis testing relative entropy for
$\varepsilon\in(0,1)$ and a sufficiently large positive integer $n$:%
\begin{equation}
D_{H}^{\varepsilon}(\rho^{\otimes n}\Vert\sigma^{\otimes n})=nD(\rho
\Vert\sigma)+\sqrt{nV(\rho\Vert\sigma)}\Phi^{-1}(\varepsilon)+O(\log n),
\label{eq:hypo-expand}%
\end{equation}
where%
\begin{equation}
\Phi(a)\equiv\frac{1}{\sqrt{2\pi}}\int_{-\infty}^{a}dx\ \exp\left(
-x^{2}/2\right)  ,\qquad\qquad\Phi^{-1}(\varepsilon)\equiv\sup\left\{
a\in\mathbb{R}\ |\ \Phi(a)\leq\varepsilon\right\}  .
\end{equation}
The equality in \eqref{eq:hypo-expand} was proven for finite-dimensional
states $\rho$ and $\sigma$ in \cite{TH12,li12}. For the case of states acting
on infinite-dimensional, separable Hilbert spaces, the inequality $\leq$ in
\eqref{eq:hypo-expand} was proven in \cite{DPR15} and \cite[Appendix~C]%
{KW17a}. In Appendix~\ref{app:hypo-expand}, we prove the inequality $\geq$ in
\eqref{eq:hypo-expand}. The proof that we detail follows the development in
\cite[Appendix~C]{DPR15} very closely, which is in turn based on
\cite[Section~3.2]{li12}.

\subsection{Communication codes}

\label{sec:code-defs}We now recall what we mean by a code for classical
communication and one for entanglement-assisted classical communication,
starting with the former. Note that classical communication was considered for
the asymptotic case in \cite{Hol98,PhysRevA.56.131}. Suppose that a channel
$\mathcal{N}_{A\rightarrow B}$ connects a sender Alice to a receiver Bob. For
positive integers $n$ and $M$, and $\varepsilon\in\left[  0,1\right]  $, an
$(n,M,\varepsilon)$ code for classical communication consists of a set
$\{\rho_{A^{n}}^{m}\}_{m\in\mathcal{M}}$ of quantum states, which are called
quantum codewords, and where $\left\vert \mathcal{M}\right\vert =M$. It also
consists of a decoding POVM $\left\{  \Lambda_{B^{n}}^{m}\right\}
_{m\in\mathcal{M}}$ and satisfies the following condition:%
\begin{equation}
\frac{1}{M}\sum_{m\in\mathcal{M}}\operatorname{Tr}\{\left(  I_{B^{n}}%
-\Lambda_{B^{n}}^{m}\right)  \mathcal{N}_{A\rightarrow B}^{\otimes n}%
(\rho_{A^{n}}^{m})\}\leq\varepsilon,
\end{equation}
which we interpret as saying that the average error probability is no larger
than $\varepsilon$, when using the quantum codewords and decoding POVM
described above. The non-asymptotic classical capacity of $\mathcal{N}%
_{A\rightarrow B}$, denoted by $C(\mathcal{N}_{A\rightarrow B},n,\varepsilon)$
is equal to the largest value of $\frac{1}{n}\log_{2}M$ (bits per channel use)
for which there exists an $\left(  n,M,\varepsilon\right)  $ code as described above.

Entanglement-assisted classical communication is defined similarly, but one
allows for Alice and Bob to share an arbitrary quantum state $\Psi_{A^{\prime
}B^{\prime}}$ before communication begins. Note that entanglement-assisted
classical communication was considered for the asymptotic case in
\cite{PhysRevLett.83.3081,ieee2002bennett,Hol01a}. For positive integers $n$
and $M$, and $\varepsilon\in\left[  0,1\right]  $, an $(n,M,\varepsilon)$ code
for entanglement-assisted classical communication consists of the resource
state $\Psi_{A^{\prime}B^{\prime}}$, a set $\{\mathcal{E}_{A^{\prime
}\rightarrow A^{n}}^{m}\}_{m\in\mathcal{M}}$ of encoding channels, where
$\left\vert \mathcal{M}\right\vert =M$. It also consists of a decoding POVM
$\left\{  \Lambda_{B^{n}B^{\prime}}^{m}\right\}  _{m\in\mathcal{M}}$ and
satisfies the following condition:%
\begin{equation}
\frac{1}{M}\sum_{m\in\mathcal{M}}\operatorname{Tr}\{\left(  I_{B^{n}B^{\prime
}}-\Lambda_{B^{n}B^{\prime}}^{m}\right)  \mathcal{N}_{A\rightarrow B}^{\otimes
n}(\mathcal{E}_{A^{\prime}\rightarrow A^{n}}^{m}(\Psi_{A^{\prime}B^{\prime}%
}))\}\leq\varepsilon,
\end{equation}
which we interpret as saying that the average error probability is no larger
than $\varepsilon$, when using the entanglement-assisted code described above.
The non-asymptotic entanglement-assisted classical capacity of $\mathcal{N}%
_{A\rightarrow B}$, denoted by $C_{\operatorname{EA}}(\mathcal{N}%
_{A\rightarrow B},n,\varepsilon)$ is equal to the largest value of $\frac
{1}{n}\log_{2}M$ (bits per channel use) for which there exists an $\left(
n,M,\varepsilon\right)  $ entanglement-assisted code as described above.

\subsection{Lower bound on second-order coding rate}

Defining the $\varepsilon$-mutual information of a bipartite state $\tau_{CD}$
as%
\begin{equation}
I_{H}^{\varepsilon}(C;D)_{\tau}\equiv D_{H}^{\varepsilon}(\tau_{CD}\Vert
\tau_{C}\otimes\tau_{D}),
\end{equation}
the following inequality was proven recently in \cite[Theorem~8]{QWW17} for
the finite-dimensional case, improving upon a prior result from \cite{DTW14}:%
\begin{equation}
C_{\operatorname{EA}}(\mathcal{N}_{A\rightarrow B},1,\varepsilon)\geq
I_{H}^{\varepsilon-\eta}(R;B)_{\zeta}-\log_{2}(4\varepsilon/\eta^{2}),
\label{eq:finite-dim-bnd}%
\end{equation}
where $\varepsilon\in(0,1)$, $\eta\in(0,\varepsilon)$, $\zeta_{RB}%
\equiv\mathcal{N}_{A\rightarrow B}(\rho_{RA})$, and $\rho_{RA}$ is a bipartite
state. The techniques employed in the proof of \cite[Theorem~8]{QWW17}\ were
position-based coding \cite{anshu2017one}\ and the Hayashi--Nagaoka inequality
\cite[Lemma~2]{HN03}. Note that the position-based coding method can be
understood as a variation of the well known and studied coding technique
called pulse position modulation \cite{verdu1990channel,eltit03}. We now
generalize the inequality in \eqref{eq:finite-dim-bnd}\ to the
infinite-dimensional case by applying Theorem~\ref{thm:main-result}, along
with position-based coding \cite{anshu2017one} and the sequential decoding
strategy from \cite{Wilde20130259}.

\begin{theorem}
\label{thm:EA-comm}Let $\mathcal{H}_{A}$, $\mathcal{H}_{B}$, and
$\mathcal{H}_{R}$ be separable Hilbert spaces. Let $\mathcal{N}_{A\rightarrow
B}$ be a quantum channel, taking $\mathcal{D}(\mathcal{H}_{A})$ to
$\mathcal{D}(\mathcal{H}_{B})$. Then the following bound holds:%
\begin{equation}
C_{\operatorname{EA}}(\mathcal{N}_{A\rightarrow B},1,\varepsilon)\geq
I_{H}^{\varepsilon-\eta}(R;B)_{\zeta}-\log_{2}(4\varepsilon/\eta^{2}),
\end{equation}
where $\varepsilon\in(0,1)$, $\eta\in(0,\varepsilon)$, $\zeta_{RB}%
\equiv\mathcal{N}_{A\rightarrow B}(\rho_{RA})$, and $\rho_{RA}\in
\mathcal{D}(\mathcal{H}_{R}\otimes\mathcal{H}_{A})$ is a bipartite state.
\end{theorem}

\begin{proof}
Let $\Lambda_{RB}$ be a measurement operator (i.e., $0\leq\Lambda_{RB}\leq
I_{RB}$) satisfying%
\begin{equation}
\operatorname{Tr}\{(I_{RB}-\Lambda_{RB})\mathcal{N}_{A\rightarrow B}(\rho
_{RA})\}\leq\varepsilon-\eta. \label{eq:cond-hypo-test}%
\end{equation}
To this operator $\Lambda_{RB}$ is associated a unitary $U_{RBP}$, defined as%
\begin{equation}
U_{RBP}\equiv\sqrt{I_{RB}-\Lambda_{RB}}\otimes\left[  |0\rangle\langle
0|_{P}+|1\rangle\langle1|_{P}\right]  +\sqrt{\Lambda_{RB}}\otimes\left[
|1\rangle\langle0|_{P}-|0\rangle\langle1|_{P}\right]  .
\end{equation}
Then defining the projectors%
\begin{align}
\Pi_{RBP}  &  \equiv U_{RBP}^{\dag}\left(  I_{RB}\otimes|1\rangle\langle
1|_{P}\right)  U_{RBP},\\
\hat{\Pi}_{RBP}  &  \equiv I_{RBP}-\Pi_{RBP}=U_{RBP}^{\dag}\left(
I_{RB}\otimes|0\rangle\langle0|_{P}\right)  U_{RBP},
\end{align}
the inequality in \eqref{eq:cond-hypo-test} and a simple calculation imply
that%
\begin{equation}
\operatorname{Tr}\{\left(  I_{RBP}-\Pi_{RBP}\right)  \mathcal{N}_{A\rightarrow
B}(\rho_{RA})\otimes|0\rangle\langle0|_{P}\}\leq\varepsilon-\eta.
\end{equation}
This kind of construction and equality is known as the Naimark extension theorem.

The position-based coding strategy then proceeds as follows. We let Alice and
Bob share $M$ copies of the resource state $\rho_{RA}$, where Bob has the $R$
systems and Alice the $A$ systems. If Alice would like to transmit message
$m\in\mathcal{M}$, then she simply selects the $m$th $A$ system, and sends it
through the channel $\mathcal{N}_{A\rightarrow B}$. The marginal state of
Bob's systems is then as follows:%
\begin{equation}
\rho_{R_{1}}\otimes\cdots\otimes\rho_{R_{m-1}}\otimes\mathcal{N}%
_{A_{m}\rightarrow B}(\rho_{R_{m}A_{m}})\otimes\rho_{R_{m+1}}\otimes
\cdots\otimes\rho_{R_{M}}.
\end{equation}
Bob then uses a sequential decoding strategy to determine which message Alice
transmitted. He introduces $M$ auxiliary probe systems in the state
$|0\rangle\langle0|$, so that Bob's overall state is now%
\begin{equation}
\omega_{R^{M}BP^{M}}^{m}\equiv\rho_{R_{1}}\otimes\cdots\otimes\rho_{R_{m-1}%
}\otimes\mathcal{N}_{A_{m}\rightarrow B}(\rho_{R_{m}A_{m}})\otimes
\rho_{R_{m+1}}\otimes\cdots\otimes\rho_{R_{M}}\otimes|0\rangle\langle
0|_{P_{1}}\otimes\cdots\otimes|0\rangle\langle0|_{P_{M}}.
\end{equation}
He then performs the binary measurements $\{\Pi_{R_{i}BP_{i}},\hat{\Pi}%
_{R_{i}BP_{i}}\}$ sequentially, in the order $i=1$, $i=2$, etc. With this
strategy, the probability that he decodes the $m$th message correctly is given
by%
\begin{equation}
\operatorname{Tr}\{\Pi_{R_{m}BP_{m}}\hat{\Pi}_{R_{m-1}BP_{m-1}}\cdots\hat{\Pi
}_{R_{1}BP_{1}}\omega_{R^{M}BP^{M}}^{m}\hat{\Pi}_{R_{1}BP_{1}}\cdots\hat{\Pi
}_{R_{m-1}BP_{m-1}}\}.
\end{equation}
Applying Theorem~\ref{thm:main-result}, we can bound the complementary (error)
probability as%
\begin{align}
p_{\text{e}}(m)  &  \equiv1-\operatorname{Tr}\{\Pi_{R_{m}BP_{m}}\hat{\Pi
}_{R_{m-1}BP_{m-1}}\cdots\hat{\Pi}_{R_{1}BP_{1}}\omega_{R^{M}BP^{M}}^{m}%
\hat{\Pi}_{R_{1}BP_{1}}\cdots\hat{\Pi}_{R_{m-1}BP_{m-1}}\}\\
&  \leq\left(  1+c\right)  \operatorname{Tr}\{\hat{\Pi}_{R_{m}BP_{m}}%
\omega_{R^{M}BP^{M}}^{m}\}+\left(  2+c+c^{-1}\right)  \sum_{i=1}%
^{m-1}\operatorname{Tr}\{\Pi_{R_{i}BP_{i}}\omega_{R^{M}BP^{M}}^{m}\}\\
&  =\left(  1+c\right)  \operatorname{Tr}\{(I_{RB}-\Lambda_{RB})\mathcal{N}%
_{A\rightarrow B}(\rho_{RA})\}\nonumber\\
&  \qquad+\left(  2+c+c^{-1}\right)  \left(  m-1\right)  \operatorname{Tr}%
\{\Lambda_{RB}\left[  \rho_{R}\otimes\mathcal{N}_{A\rightarrow B}(\rho
_{A})\right]  \}\\
&  \leq\left(  1+c\right)  \left(  \varepsilon-\eta\right)  +\left(
2+c+c^{-1}\right)  M\operatorname{Tr}\{\Lambda_{RB}\left[  \rho_{R}%
\otimes\mathcal{N}_{A\rightarrow B}(\rho_{A})\right]  \},
\end{align}
where $c>0$. Since the whole development above holds for all measurement
operators $\Lambda_{RB}$ satisfying \eqref{eq:cond-hypo-test}, we can take an
infimum over all of them to obtain the following uniform bound on the error
probability when sending an arbitrary message $m\in\mathcal{M}$:%
\begin{equation}
p_{\text{e}}(m)\leq\left(  1+c\right)  \left(  \varepsilon-\eta\right)
+\left(  2+c+c^{-1}\right)  M2^{-I_{H}^{\varepsilon-\eta}(R;B)_{\zeta}}.
\end{equation}
Picking $c=\eta/(2\varepsilon-\eta)$ and taking $M$ such that%
\begin{equation}
\log_{2}M=I_{H}^{\varepsilon-\eta}(R;B)_{\zeta}-\log_{2}(4\varepsilon/\eta
^{2}) \label{eq:bits-ea-code}%
\end{equation}
then implies that $p_{\text{e}}(m)\leq\varepsilon$ for all $m\in\mathcal{M}$.
Since we have shown the existence of a $(1,M,\varepsilon)$
entanglement-assisted code, where $M$ satisfies \eqref{eq:bits-ea-code}, this
concludes the proof.
\end{proof}

\begin{remark}
It is worthwhile to note that the code above has an error probability less
than $\varepsilon$ for every message $m \in\mathcal{M}$, and so the error
criterion is maximal error probability and not just average error probability.
\end{remark}

The above result also implies rates that are achievable for unassisted
classical communication, by combining Theorem~\ref{thm:EA-comm}\ with an
analysis nearly identical to that given in \cite[Section~3.3]%
{wilde2017position}. In particular, we could allow Alice and Bob to share many
copies of the following classical--quantum state before communication begins:%
\begin{equation}
\rho_{XA}\equiv\sum_{x\in\mathcal{X}}p(x)|x\rangle\langle x|_{X}\otimes
\rho_{A}^{x}, \label{eq:cq-state}%
\end{equation}
where $\mathcal{H}_{X}$ and $\mathcal{H}_{A}$ are separable Hilbert spaces,
$\mathcal{X}$ is a countable index set, $\{p(x)\}_{x\in\mathcal{X}}$ is a
probability distribution, $\{|x\rangle_{X}\}_{x\in\mathcal{X}}$ is a set of
orthonormal states, and $\{\rho_{A}^{x}\}_{x\in\mathcal{X}}$ is a set of
states. This state then plays the role of the resource state $\rho_{RA}$ in
the proof of Theorem~\ref{thm:EA-comm}. However, the above state is a
classical--quantum state, and as such, the code can be derandomized.
Specifically, to do so, we can employ the analysis given in \cite[Section~3.3]%
{wilde2017position}, but replacing the square-root measurement there with the
sequential decoding strategy.\ This leads to the following result, which
generalizes one of the main results of \cite{WR12} to the infinite-dimensional case:

\begin{corollary}
\label{thm:unassisted-comm}Let $\mathcal{H}_{A}$, $\mathcal{H}_{B}$, and
$\mathcal{H}_{X}$ be separable Hilbert spaces. Let $\mathcal{N}_{A\rightarrow
B}$ be a quantum channel, taking $\mathcal{D}(\mathcal{H}_{A})$ to
$\mathcal{D}(\mathcal{H}_{B})$. Then the following bound holds:%
\begin{equation}
C(\mathcal{N}_{A\rightarrow B},1,\varepsilon)\geq I_{H}^{\varepsilon-\eta
}(X;B)_{\zeta}-\log_{2}(4\varepsilon/\eta^{2}),
\end{equation}
where $\varepsilon\in(0,1)$, $\eta\in(0,\varepsilon)$, $\zeta_{XB}%
\equiv\mathcal{N}_{A\rightarrow B}(\rho_{XA})$, and $\rho_{XA}\in
\mathcal{D}(\mathcal{H}_{X}\otimes\mathcal{H}_{A})$ is a bipartite,
classical--quantum state of the form in \eqref{eq:cq-state}.
\end{corollary}

\subsection{Energy constraints}

It is common in the theory of communication over infinite-dimensional channels
\cite{H03,H04,H12book}\ to impose energy constraints on the space of inputs.
If we do not so, then the capacities can be infinite. The definitions of these
energy-constrained non-asymptotic capacities are the same as we discussed
previously, except that we impose energy constraints on the channel input states.

Before defining them, let us first recall the notion of an energy observable
\cite{H12book,HS12}:

\begin{definition}
[Energy Observable]For a Hilbert space $\mathcal{H}$, let $G\in\mathcal{L}%
_{+}(\mathcal{H})$ denote a positive semi-definite operator, defined in terms
of its action on a vector $|\psi\rangle$ as
\begin{equation}
G|\psi\rangle=\sum_{j=1}^{\infty}g_{j}|e_{j}\rangle\langle e_{j}|\psi\rangle,
\label{eq:energyObservable}%
\end{equation}
for $|\psi\rangle$ such that $\sum_{j=1}^{\infty}g_{j}|\langle e_{j}%
|\psi\rangle|^{2}<\infty$. In the above, $\{|e_{j}\rangle\}_{j}$ is an
orthonormal basis and $\{g_{j}\}_{j}$ is a sequence of non-negative, real
numbers. Then $\{|e_{j}\rangle\}_{j}$ is an eigenbasis for $G$ with
corresponding eigenvalues $\{g_{j}\}_{j}$.
\end{definition}

\begin{definition}
The $n$th extension $\overline{G}_{n}$\ of an energy observable $G$ is defined
as%
\begin{equation}
\overline{G}_{n}\equiv\frac{1}{n}\left[  G\otimes I\otimes\cdots\otimes
I+\cdots+I\otimes\cdots\otimes I\otimes G\right]  ,
\end{equation}
where $n$ is the number of factors in each tensor product above.
\end{definition}

Then the non-asymptotic, energy-constrained classical capacity $C(\mathcal{N}%
_{A\rightarrow B},G,P,n,\varepsilon)$ is defined exactly as it was previously
in Section~\ref{sec:code-defs}, except that we demand that%
\begin{equation}
\frac{1}{M}\sum_{m\in\mathcal{M}}\operatorname{Tr}\{\overline{G}_{n}%
\rho_{A^{n}}^{m}\}\leq P,
\end{equation}
for a real $P\in\lbrack0,\infty)$. Similarly, the non-asymptotic,
energy-constrained entanglement-assisted classical capacity
$C_{\operatorname{EA}}(\mathcal{N}_{A\rightarrow B},G,P,n,\varepsilon)$ is
defined exactly as it was previously in Section~\ref{sec:code-defs}, except
that we demand that%
\begin{equation}
\frac{1}{M}\sum_{m\in\mathcal{M}}\operatorname{Tr}\{(\overline{G}_{n}\otimes
I_{B^{\prime}})\mathcal{E}_{A^{\prime}\rightarrow A^{n}}^{m}(\Psi_{A^{\prime
}B^{\prime}})\}\leq P.
\end{equation}
One could alternatively demand that the energy constraint hold for every
codeword, not just on average. Note that we recover the usual notion of
capacity (unconstrained) by taking $G=I$ and setting $P=1$.

An advantage of the approach given in the proof of Theorem~\ref{thm:EA-comm}
is that we easily obtain a lower bound on the second-order coding rate for
energy-constrained entanglement-assisted classical communication over a
quantum channel:

\begin{theorem}
\label{thm:EA-comm-EC}Let $\mathcal{H}_{A}$, $\mathcal{H}_{B}$, and
$\mathcal{H}_{R}$ be separable Hilbert spaces. Let $\varepsilon\in(0,1)$. Let
$G$ be an energy observable, and let $P\in\lbrack0,\infty)$. Let
$\mathcal{N}_{A\rightarrow B}$ be a quantum channel, taking $\mathcal{D}%
(\mathcal{H}_{A})$ to $\mathcal{D}(\mathcal{H}_{B})$. Then the following bound
holds:%
\begin{equation}
C_{\operatorname{EA}}(\mathcal{N}_{A\rightarrow B},G,P,n,\varepsilon)\geq
I(R;B)_{\zeta}+\sqrt{\frac{1}{n}V(R;B)_{\zeta}}\Phi^{-1}(\varepsilon
)+O\!\left(  \frac{1}{n}\log n\right)  ,
\end{equation}
where $\zeta_{RB}\equiv\mathcal{N}_{A\rightarrow B}(\rho_{RA})$ and $\rho
_{RA}\in\mathcal{D}(\mathcal{H}_{R}\otimes\mathcal{H}_{A})$ is a bipartite
state such that%
\begin{equation}
I(R;B)_{\zeta},\ V(R;B)_{\zeta},\ T(R;B)_{\zeta}<\infty,\qquad V(R;B)_{\zeta
}>0,
\end{equation}
and $\operatorname{Tr}\{G\rho_{A}\}\leq P $. In the above, we have the mutual
information, mutual information variance, and another quantity:%
\begin{equation}
I(R;B)_{\zeta} \equiv D(\zeta_{RB}\Vert\zeta_{R}\otimes\zeta_{B}),\quad
V(R;B)_{\zeta} \equiv V(\zeta_{RB}\Vert\zeta_{R}\otimes\zeta_{B}),\quad
T(R;B)_{\zeta} \equiv T(\zeta_{RB}\Vert\zeta_{R}\otimes\zeta_{B}).
\end{equation}

\end{theorem}

\begin{proof}
Let $\zeta_{RA}$ be a state satisfying the conditions stated above. Then the
claim follows by applying Theorem~\ref{thm:EA-comm}, picking $\eta=1/\sqrt{n}%
$, and invoking the expansion in \eqref{eq:hypo-expand}.
\end{proof}

\bigskip

The proof given above is quite simple once all of the relevant components are
in place (namely, the quantum union bound from Theorem~\ref{thm:main-result},
position-based coding \cite{anshu2017one}, and the expansion in
\eqref{eq:hypo-expand}). This is to be contrasted with the approach taken in
\cite{H03,H04}, in which the energy-constrained entanglement-assisted
classical capacity was identified. Not only can we argue to have a simpler
approach for the achievability part, but our method also delivers a lower
bound on the second-order coding rate. An important open question remaining
however is to determine whether this lower bound on the second-order coding
rate is tight. To our knowledge, this tightness has only been shown for
finite-dimensional channels that are covariant \cite{DTW14}.

We note here that the bound in Theorem~\ref{thm:EA-comm-EC} applies to the
practically relevant case of bosonic Gaussian channels \cite{S17}. The
energy-constrained entanglement-assisted classical capacity of these channels
was identified in \cite{HW01,H03,GLMS03,H04}. The additive-noise, thermal, and
amplifier channels are of major interest for applications, as stressed in
\cite{HG12,H12book}. It is known that the energy-constrained,
entanglement-assisted capacity formula for these channels is achieved by a
two-mode squeezed vacuum state, whose reduction to the channel input system
has an average photon number meeting the desired photon number constraint.
Thus, we could evaluate the lower bound in Theorem~\ref{thm:EA-comm-EC} by
taking $\rho_{RA}$ therein to be the two-mode squeezed vacuum and then
applying the formula from \cite{BLTW16}\ to evaluate the mutual information
variance $V(R;B)_{\zeta}$, while noting that the quantity $T(R;B)_{\zeta}$ is
finite for any finite-energy state, as proven in \cite[Appendix~D]{KW17a}.

We end this section on a different note, by remarking that the same argument
as above gives a non-trivial lower bound on the second-order coding rate for
energy-constrained classical communication with randomness assistance.
However, it remains open to determine whether this rate is achievable without
the assistance of randomness. It is also open to extend the result to a
continuous (uncountable)\ index set $\mathcal{X}$. We suspect that these
extensions should be possible but leave it for future work.

\section{Conclusion}

\label{sec:concl}In this paper, we proved Theorem~\ref{thm:main-result}, which
improves Gao's quantum union bound to include a tunable parameter $c>0$ that
plays a role similar to the tunable parameter available in the
Hayashi--Nagaoka inequality from \cite[Lemma~2]{HN03}. An advantage of the
proof of Theorem~\ref{thm:main-result} is that it is elementary, relying only
on basic properties of projectors, the Pythagorean theorem, and the
Cauchy--Schwarz inequality. Due to our improvement, the quantum union bound
can now be employed in a wider variety of scenarios. As an example
application, we showed how to achieve a lower bound on the second-order coding
rate for classical communication over a quantum channel by employing a
sequential decoding strategy.

For future directions, we think it would be interesting to determine whether
the improved bound in Theorem~\ref{thm:main-result}\ would find application in
areas
including quantum algorithms \cite{MW16,A16,HYM17}, quantum complexity theory
\cite{A16,WHKN18}, and Hamiltonian complexity theory \cite{AAV16,KL18}. We
also wonder whether Theorem~\ref{thm:main-result} could be useful outside of
quantum information, for example in the analysis of projection algorithms.

\appendix

\section{Proof of the inequality $\geq$ in Eq.~\eqref{eq:hypo-expand}}

\label{app:hypo-expand}

The goal of this appendix is to prove the inequality $\geq$ in
~\eqref{eq:hypo-expand}. The proof follows the development in Appendix~C of
\cite{DPR15} very closely, which is in turn based on \cite[Section~3.2]{li12}.

Consider quantum states $\rho$ and $\sigma$ acting on a separable Hilbert
space $\mathcal{H}$, with spectral decompositions as given in
\eqref{eq:spec-decomp-rho-sig}. Observe that each $P_{x}$ is
finite-dimensional, as a consequence of the fact that $\rho$ is trace class.
Indeed, were it not the case, then $\operatorname{Tr}\{P_{x}\}$ would be
infinite, and $\rho$ could thus not be trace class. By the same reasoning,
each $Q_{y}$ is finite-dimensional.

By defining a random variable $Z$ taking values $\log_{2}(\lambda_{x}/\mu
_{y})$ with probability $p(x,y)\equiv\lambda_{x}\operatorname{Tr}\{P_{x}%
Q_{y}\}$, observe that%
\begin{equation}
D(\rho\Vert\sigma)=\mathbb{E}\{Z\},\qquad V(\rho\Vert\sigma)=\text{Var}%
\{Z\},\qquad T(\rho\Vert\sigma)=\mathbb{E}\left\{  \left\vert Z-\mathbb{E}%
\{Z\}\right\vert ^{3}\right\}  , \label{eq:mean-of-Z}%
\end{equation}
where $D(\rho\Vert\sigma)$, $V(\rho\Vert\sigma)$, and $T(\rho\Vert\sigma)$ are
defined in \eqref{eq:rel-ent-def}-\eqref{eq:rel-ent-T-def}.

\begin{lemma}
\label{lem:ke-li-sep-hilb} Let $\rho$ and $\sigma$ denote states acting on a
separable Hilbert space $\mathcal{H}$. Let $L>0$. Then there exists a
measurement operator $T_{L}$ (i.e., $0\leq T_{L}\leq I$) such that%
\begin{equation}
\operatorname{Tr}\{T_{L}\rho\} \geq\Pr\{Z\geq\log_{2} L\},\qquad
\qquad\operatorname{Tr}\{T_{L}\sigma\} \leq\frac{1}{L},
\end{equation}
where $Z$ is the random variable defined just before \eqref{eq:mean-of-Z}.
\end{lemma}

\begin{proof}
Let us define the positive semi-definite operator $\widetilde{T}_{L}$ as%
\begin{equation}
\widetilde{T}_{L}\equiv\sum_{x,y:L\leq\lambda_{x}/\mu_{y}}Q_{y}P_{x}Q_{y}.
\end{equation}
By inspection, this operator is positive semi-definite. The measurement
operator $T_{L}$ is then defined to be the projection onto the support of
$\widetilde{T}_{L}$. Let $|\psi\rangle$ be a unit vector such that $P_{x}%
|\psi\rangle=|\psi\rangle$ for some $x$. It follows that $|\psi\rangle
\langle\psi|\leq P_{x}$. Then, for any $\mu_{y}$ such that $L\leq\lambda
_{x}/\mu_{y}$, we have, from the definition of $\widetilde{T}_{L}$, that
$|\psi\rangle\langle\psi|\leq\widetilde{T}_{L}$. This in turn implies that
$Q_{y}|\psi\rangle\in\operatorname{supp}(\widetilde{T}_{L})$. From this, we
then conclude that%
\begin{equation}
\frac{Q_{y}|\psi\rangle\langle\psi|Q_{y}}{\left\Vert Q_{y}|\psi\rangle
\right\Vert ^{2}}\leq T_{L}.
\end{equation}
Furthermore, we have that $\{Q_{y}|\psi\rangle\}_{y}$ forms a family of
orthogonal vectors. Then the following inequality holds%
\begin{equation}
\sum_{y:L\leq\lambda_{x}/\mu_{y}}\frac{Q_{y}|\psi\rangle\langle\psi|Q_{y}%
}{\left\Vert Q_{y}|\psi\rangle\right\Vert ^{2}}\leq T_{L}.
\end{equation}
From this, we conclude that%
\begin{align}
\operatorname{Tr}\{T_{L}|\psi\rangle\langle\psi|\} &  \geq\operatorname{Tr}%
\left\{  \sum_{y:L\leq\lambda_{x}/\mu_{y}}\frac{Q_{y}|\psi\rangle\langle
\psi|Q_{y}}{\left\Vert Q_{y}|\psi\rangle\right\Vert ^{2}}|\psi\rangle
\langle\psi|\right\}  =\sum_{y:L\leq\lambda_{x}/\mu_{y}}\operatorname{Tr}%
\left\{  \frac{Q_{y}|\psi\rangle\langle\psi|Q_{y}}{\left\Vert Q_{y}%
|\psi\rangle\right\Vert ^{2}}|\psi\rangle\langle\psi|\right\}
\label{eq:TL-on-psi-1}\\
&  =\sum_{y:L\leq\lambda_{x}/\mu_{y}}\frac{\left\Vert Q_{y}|\psi
\rangle\right\Vert ^{4}}{\left\Vert Q_{y}|\psi\rangle\right\Vert ^{2}}%
=\sum_{y:L\leq\lambda_{x}/\mu_{y}}\left\Vert Q_{y}|\psi\rangle\right\Vert
^{2}=\sum_{y:L\leq\lambda_{x}/\mu_{y}}\operatorname{Tr}\{Q_{y}|\psi
\rangle\langle\psi|Q_{y}\}.\label{eq:TL-on-psi-4}%
\end{align}
Now, recall that $P_{x}$ is a finite-dimensional projector for each $x$. (As
stated above, if $P_{x}$ were not finite-dimensional, then this would
contradict the assumption that $\rho$ is trace class.) Furthermore, we can
write it as $P_{x}=\sum_{l=1}^{\operatorname{Tr}\{P_{x}\}}|\psi_{x,l}%
\rangle\langle\psi_{x,l}|$, for some orthonormal set $\{|\psi_{x,l}%
\rangle\}_{l=1}^{\operatorname{Tr}\{P_{x}\}}$. Then the development in
\eqref{eq:TL-on-psi-1}--\eqref{eq:TL-on-psi-4}\ implies that%
\begin{equation}
\operatorname{Tr}\{T_{L}P_{x}\}\geq\sum_{y:L\leq\lambda_{x}/\mu_{y}%
}\operatorname{Tr}\{Q_{y}P_{x}Q_{y}\}=\sum_{y:L\leq\lambda_{x}/\mu_{y}%
}\operatorname{Tr}\{Q_{y}P_{x}\}.
\end{equation}
We can then use this to conclude that%
\begin{align}
\operatorname{Tr}\{T_{L}\rho\} &  =\sum_{x}\lambda_{x}\operatorname{Tr}%
\{T_{L}P_{x}\}\geq\sum_{x,y:L\leq\lambda_{x}/\mu_{y}}\lambda_{x}%
\operatorname{Tr}\{Q_{y}P_{x}\}\\
&  =\sum_{x,y:\log_{2}L\leq\log_{2}(\lambda_{x}/\mu_{y})}\lambda
_{x}\operatorname{Tr}\{Q_{y}P_{x}\}=\Pr\{Z\geq\log_{2}L\},
\end{align}
where the second equality uses the fact that $\log_{2}:(0,\infty
)\rightarrow(-\infty,\infty)$ is invertible, and the last line follows from
the definition of the random variable~$Z$, given just before \eqref{eq:mean-of-Z}.

What remains is to place an upper bound on $\operatorname{Tr}\{T_{L}\sigma\}$.
Observe that for all $x$ and $y$, the following equivalence holds%
\begin{equation}
\operatorname{ran}(Q_{y}P_{x})=\operatorname{supp}(Q_{y}P_{x}Q_{y}),
\end{equation}
where $\operatorname{ran}$ denotes the range of an operator. For some $x$,
define the following subspace:%
\begin{equation}
\widetilde{S}_{x}\equiv\bigvee\limits_{y:L\leq\lambda_{x}/\mu_{y}%
}\operatorname{ran}(Q_{y}P_{x})=\bigvee\limits_{y:L\leq\lambda_{x}/\mu_{y}%
}\operatorname{supp}(Q_{y}P_{x}Q_{y}),
\end{equation}
where the operation $\vee$ realizes a space formed as the union of subspaces.
Due to the fact that $P_{x}$ is finite-dimensional, it follows that the
subspace $\widetilde{S}_{x}$ is finite-dimensional. We now employ a
Gram--Schmidt orthogonalization procedure for these subspaces. First order the
eigenvalues of $\rho$ as $\lambda_{1}>\lambda_{2}>\ldots$. Now define a family
$\{S_{x}\}_{x}$ of subspaces of the whole Hilbert space $\mathcal{H}$ as%
\begin{equation}
S_{1}\equiv\widetilde{S}_{1},\qquad\qquad S_{x}\equiv\left(  \bigvee
\limits_{i=1}^{x}\widetilde{S}_{i}\right)  \wedge\left(  \bigvee
\limits_{i=1}^{x-1}\widetilde{S}_{i}\right)  ^{\perp}\text{ for }x\geq2,
\end{equation}
where the operation $\wedge$ corresponds to the intersection of subspaces. The
subspaces $S_{x}$ are mutually orthogonal by construction. Furthermore, by the
procedure given above, the following holds for any positive integer $w\geq1$:%
\begin{equation}
\bigvee\limits_{x=1}^{w}\widetilde{S}_{x}=\bigvee\limits_{x=1}^{w}S_{x}.
\end{equation}
By definition, $T_{L}$ is the projection onto the following subspace:%
\begin{equation}
\bigvee\limits_{x,y:L\leq\lambda_{x}/\mu_{y}}\operatorname{supp}(Q_{y}%
P_{x}Q_{y})=\bigvee\limits_{y:L\leq\lambda_{x}/\mu_{y}}\operatorname{ran}%
(Q_{y}P_{x})=\bigvee\limits_{x}\widetilde{S}_{x}=\bigvee\limits_{x}%
S_{x}=\bigoplus\limits_{x}S_{x}.
\end{equation}
Thus, it follows that $T_{L}$ can be written as $T_{L}=\sum_{x}P_{S_{x}}$,
where $P_{S_{x}}$ is the projection onto $S_{x}$. By the procedure given
above, we have that $S_{x}\subseteq\widetilde{S}_{x}$, and from the definition
of $\widetilde{S}_{x}$, we find that%
\begin{equation}
\operatorname{Tr}\{P_{S_{x}}\}\leq\operatorname{Tr}\{P_{\widetilde{S}_{x}%
}\}\leq\operatorname{Tr}\{P_{x}\}.\label{eq:proj-ordering-GS}%
\end{equation}
We then find that%
\begin{align}
\operatorname{Tr}\{T_{L}\sigma\} &  =\sum_{y,x}\mu_{y}\operatorname{Tr}%
\{Q_{y}P_{S_{x}}\}=\sum_{y,x:L\leq\lambda_{x}/\mu_{y}}\mu_{y}\operatorname{Tr}%
\{Q_{y}P_{S_{x}}\}\\
&  \leq\frac{1}{L}\sum_{y,x:L\leq\lambda_{x}/\mu_{y}}\lambda_{x}%
\operatorname{Tr}\{Q_{y}P_{S_{x}}\}\leq\frac{1}{L}\sum_{y,x}\lambda
_{x}\operatorname{Tr}\{Q_{y}P_{S_{x}}\}\\
&  =\frac{1}{L}\sum_{x}\lambda_{x}\operatorname{Tr}\{P_{S_{x}}\}\leq\frac
{1}{L}\sum_{x}\lambda_{x}\operatorname{Tr}\{P_{x}\}=\frac{1}{L}%
\operatorname{Tr}\{\rho\}=\frac{1}{L}.
\end{align}
In the above, the second equality follows because $Q_{y}P_{S_{x}}=0$ unless
$L\leq\lambda_{x}/\mu_{y}$ (from the definition of the space $S_{x}$ and the
fact that $S_{x}\subseteq\widetilde{S}_{x}$). The third equality follows from
$\sum_{y}Q_{y}=I$, and the third inequality follows from~\eqref{eq:proj-ordering-GS}.
\end{proof}

\bigskip We now apply the above lemma to the i.i.d.~states $\rho^{\otimes n}$
and $\sigma^{\otimes n}$, with spectral decompositions%
\begin{equation}
\rho^{\otimes n}=\sum_{x^{n}}\lambda_{x^{n}}P_{x^{n}},\qquad\sigma^{\otimes
n}=\sum_{y^{n}}\mu_{y^{n}}Q_{y^{n}},
\end{equation}
where $x^{n}=\left(  x_{1},\ldots,x_{n}\right)  $, $y^{n}=\left(  y_{1}%
,\ldots,y_{n}\right)  $, $\lambda_{x^{n}}=\lambda_{x_{1}}\times\cdots
\times\lambda_{x_{n}}$, $\mu_{y^{n}}=\mu_{y_{1}}\times\cdots\times\mu_{y_{n}}%
$, $P_{x^{n}}=P_{x_{1}}\otimes\cdots\otimes P_{x_{n}}$, and $Q_{y^{n}%
}=Q_{y_{1}}\otimes\cdots\otimes Q_{y_{n}}$. Then the i.i.d.~random sequence
$Z^{n}\equiv(Z_{1},\ldots,Z_{n})$ takes the values%
\begin{equation}
\log_{2}\!\left(  \frac{\lambda_{x^{n}}}{\mu_{y^{n}}}\right)  =\sum_{i=1}%
^{n}\log_{2}\!\left(  \frac{\lambda_{x_{i}}}{\mu_{y_{i}}}\right)  ,
\label{eq:rand-Z-sequence}%
\end{equation}
with probability%
\begin{equation}
p(x^{n},y^{n})=\lambda_{x^{n}}\operatorname{Tr}\{P_{x^{n}}Q_{y^{n}}%
\}=\prod\limits_{i=1}^{n}\lambda_{x_{i}}\operatorname{Tr}\{P_{x_{i}}Q_{y_{i}%
}\}. \label{eq:rand-Z-sequence-2}%
\end{equation}

The Berry--Essen theorem \cite{KS10,S11a}\ states that if $A_{1}$, \ldots,
$A_{n}$ are i.i.d.~random variables such that $\mathbb{E}\{A_{1}\}=0$,
$\mathbb{E}\{\left\vert A_{1}\right\vert ^{2}\}\equiv\tau^{2}\in(0,\infty)$,
and $\mathbb{E}\{\left\vert A_{1}\right\vert ^{3}\}\equiv\omega<\infty$, then%
\begin{equation}
\left\vert \Pr\{B_{n}\sqrt{n}/\tau\leq x\}-\Phi(x)\right\vert \leq
\frac{C\omega}{\tau^{3}\sqrt{n}},
\end{equation}
where $x\in\mathbb{R}$, $\Phi(x)\equiv\left[  2\pi\right]  ^{-1/2}%
\int_{-\infty}^{x}dy\exp(-y^{2}/2)$, $B_{n}\equiv\frac{1}{n}\sum_{i=1}%
^{n}A_{i}$, and $C$ is the Berry--Esseen constant satisfying $0.40973\leq
C\leq0.4784$.

\begin{proposition}
Let $\rho$ and $\sigma$ denote states acting on a separable Hilbert
space~$\mathcal{H}$. Suppose that $D(\rho\Vert\sigma),V(\rho\Vert
\sigma),T(\rho\Vert\sigma)<\infty$ and $V(\rho\Vert\sigma)>0$. Suppose $n$ is
sufficiently large such that $\varepsilon-C\cdot T(\rho\Vert\sigma
)/\sqrt{n\left[  V(\rho\Vert\sigma)\right]  ^{3}}>0$. Then%
\begin{align}
D_{H}^{\varepsilon}(\rho^{\otimes n}\Vert\sigma^{\otimes n}) &  \geq
nD(\rho\Vert\sigma)+\sqrt{nV(\rho\Vert\sigma)}\Phi^{-1}\!\left(
\varepsilon-C\cdot T(\rho\Vert\sigma)/\sqrt{n\left[  V(\rho\Vert
\sigma)\right]  ^{3}}\right)  \label{eq:2nd-order-expansion-final}\\
&  =nD(\rho\Vert\sigma)+\sqrt{nV(\rho\Vert\sigma)}\Phi^{-1}\!\left(
\varepsilon\right)  +O(1).
\end{align}

\end{proposition}

\begin{proof}
Applying the Berry--Esseen theorem to the random sequence $Z_{1}-D(\rho
\Vert\sigma)$, \ldots, $Z_{n}-D(\rho\Vert\sigma)$, with $Z_{i}$ defined in
\eqref{eq:rand-Z-sequence}--\eqref{eq:rand-Z-sequence-2}, we find that%
\begin{equation}
\left\vert \Pr\left\{  \overline{Z^{n}}\sqrt{\frac{n}{V(\rho\Vert\sigma)}}\leq
x\right\}  -\Phi(x)\right\vert \leq C\cdot T(\rho\Vert\sigma)/\sqrt{n\left[
V(\rho\Vert\sigma)\right]  ^{3}},
\end{equation}
where $\overline{Z^{n}}\equiv\frac{1}{n}\sum_{i=1}^{n}\left[  Z_{i}%
-D(\rho\Vert\sigma)\right]  $, which implies that%
\begin{equation}
\Pr\left\{  \sum_{i=1}^{n}Z_{i}\leq nD(\rho\Vert\sigma)+x\sqrt{nV(\rho
\Vert\sigma)}\right\}  \leq\Phi(x)+C\cdot T(\rho\Vert\sigma)/\sqrt{n\left[
V(\rho\Vert\sigma)\right]  ^{3}}.
\end{equation}
Picking $x=\Phi^{-1}\!\left(  \varepsilon-C\cdot T(\rho\Vert\sigma
)/\sqrt{n\left[  V(\rho\Vert\sigma)\right]  ^{3}}\right)  $, this becomes%
\begin{equation}
\Pr\left\{  \sum_{i=1}^{n}Z_{i}\leq nD(\rho\Vert\sigma)+\sqrt{nV(\rho
\Vert\sigma)}\Phi^{-1}\!\left(  \varepsilon-C\cdot T(\rho\Vert\sigma
)/\sqrt{n\left[  V(\rho\Vert\sigma)\right]  ^{3}}\right)  \right\}
\leq\varepsilon.
\end{equation}
Choosing $L$ such that
\begin{equation}
\log_{2}L=nD(\rho\Vert\sigma)+\sqrt{nV(\rho\Vert\sigma)}\Phi^{-1}\!\left(
\varepsilon-C\cdot T(\rho\Vert\sigma)/\sqrt{n\left[  V(\rho\Vert
\sigma)\right]  ^{3}}\right)
\end{equation}
and applying Lemma~\ref{lem:ke-li-sep-hilb}, we find that%
\[
\operatorname{Tr}\{T^{n}\rho^{\otimes n}\}\geq\Pr\left\{  \sum_{i=1}^{n}%
Z_{i}\geq\log_{2}L\right\}  =1-\Pr\left\{  \sum_{i=1}^{n}Z_{i}\leq\log
_{2}L\right\}  \geq1-\varepsilon,
\]
while%
\begin{equation}
\operatorname{Tr}\{T^{n}\sigma^{\otimes n}\}\leq\frac{1}{L}=e^{-\left[
nD(\rho\Vert\sigma)+\sqrt{nV(\rho\Vert\sigma)}\Phi^{-1}\!\left(
\varepsilon-C\cdot T(\rho\Vert\sigma)/\sqrt{n\left[  V(\rho\Vert
\sigma)\right]  ^{3}}\right)  \right]  }.
\end{equation}
This implies that%
\begin{equation}
-\log_{2}\operatorname{Tr}\{T^{n}\sigma^{\otimes n}\}\geq nD(\rho\Vert
\sigma)+\sqrt{nV(\rho\Vert\sigma)}\Phi^{-1}\!\left(  \varepsilon-C\cdot
T(\rho\Vert\sigma)/\sqrt{n\left[  V(\rho\Vert\sigma)\right]  ^{3}}\right)  .
\end{equation}
Since $D_{H}^{\varepsilon}(\rho^{\otimes n}\Vert\sigma^{\otimes n})$ involves
an optimization over all possible measurement operators $T^{n}$ satisfying
$\operatorname{Tr}\{T^{n}\rho^{\otimes n}\}\geq1-\varepsilon$, we conclude the
bound in \eqref{eq:2nd-order-expansion-final}. The equality after
\eqref{eq:2nd-order-expansion-final} follows from expanding $\Phi^{-1}$ at the
point $\varepsilon$ using Lagrange's mean value theorem.
\end{proof}

\bigskip\noindent\textbf{Ethics statement}. This work did not involve any
active collection of human data, and it did not involve animals.

\bigskip\noindent\textbf{Data accessibility statement}. This work does not
have any experimental data.

\bigskip\noindent\textbf{Competing interests statement}. We have no competing interests.

\bigskip\noindent\textbf{Authors' contributions}. All authors contributed
equally toward proving the main results and writing the paper.

\bigskip\noindent\textbf{Acknowledgements}. MMW\ is grateful to SM\ for
hosting him for a research visit to University of Camerino in June and July,
2016, during which this research project was initiated. SM is grateful to MMW
for hosting him for a research visit to Louisiana State University at Baton
Rouge in March, 2018, during which this research project was completed.

\bigskip\noindent\textbf{Funding}. MMW acknowledges support from the US
National Science Foundation under grant no.~1714215.

\bibliographystyle{unsrt}
\bibliography{Ref}

\end{document}